\documentclass[11pt]{article}

\usepackage[margin=1in]{geometry}
\usepackage{amsfonts,amssymb,amsmath,amsthm,mathtools}
\usepackage{thmtools,thm-restate}

\usepackage{graphicx,microtype,latexsym,lpic,bm,xspace,booktabs,wrapfig,array}
\usepackage[bf]{caption}

\usepackage{algorithmicx,algpseudocode}

\usepackage{dsfont}
\usepackage[greek,english]{babel}

\usepackage{enumitem}
\setlist[itemize]{noitemsep,label=$-$}
\setlist[enumerate]{noitemsep,label=\itshape(\arabic*)}

\usepackage{hyperref}
\addto\extrasenglish{}
\addto\extrasenglish{}
\newcommand{\appref}[1]{\hyperref[#1]{Appendix~\ref{#1}}}

\usepackage[usenames,dvipsnames]{xcolor}
\usepackage[framemethod=tikz]{mdframed}
\mdfsetup{
	nobreak=true,
	linewidth=.5,
	innerleftmargin=\parindent,
	innerrightmargin=20pt,
	innertopmargin=12pt,
	innerbottommargin=14pt,
	skipabove=0,skipbelow=0
}

\usepackage{tikz}
\usetikzlibrary{arrows.meta} 

\hypersetup{
	colorlinks=true,
	linkcolor=Sepia,
	citecolor=Sepia,
	filecolor=Sepia,
	urlcolor=Sepia
}

\declaretheorem[name=Theorem]{theorem}
\declaretheorem[name=Lemma,sibling=theorem]{lemma}
\declaretheorem[name=Claim,sibling=theorem]{claim}

\declaretheorem[name=Definition,style=definition]{definition}

\renewcommand{\H}{\mathbf{H}}
\renewcommand{\Pr}{\mathbf{Pr}}
\newcommand{\Ex}{\mathbf{E}}
\newcommand{\Hmin}{\H_\infty}
\newcommand{\Dmin}{\mathbf{D}_\infty}

\DeclareMathOperator{\poly}{poly}
\DeclareMathOperator{\free}{free}
\DeclareMathOperator{\fix}{fix}
\DeclareMathOperator{\supp}{supp}

\newcommand{\smallFunction}[2]{\newcommand{#1}{{\textsc{#2}}}}

\smallFunction{\Or}{Or}
\smallFunction{\Xor}{Xor}
\smallFunction{\Ind}{Ind}
\smallFunction{\Forr}{Forrelation}
\smallFunction{\Eol}{End-of-Line}

\newcommand{\newclass}[2]{\newcommand{#1}{{\text{\upshape\sffamily #2}}\xspace}}
\renewcommand{\P}{\text{\upshape\sffamily P}\xspace}
\newclass{\BPP}{BPP}
\newclass{\NP}{NP}
\newclass{\RP}{RP}
\newclass{\ZPP}{ZPP}
\newclass{\BQP}{BQP}
\newclass{\PPAD}{PPAD}
\newclass{\WAPP}{WAPP}

\newcommand{\newsftext}[2]{\newcommand{#1}{{\text{\upshape\sffamily #2}}}}
\newsftext{\dt}{dt}
\newsftext{\cc}{cc}

\newcommand{\BPPcc}{\BPP^{\cc}}
\newcommand{\BPPdt}{\BPP^{\dt}}
\newcommand{\BQPcc}{\BQP^{\cc}}
\newcommand{\BQPdt}{\BQP^{\dt}}

\renewcommand{\b}{\bm{b}}
\renewcommand{\c}{\bm{c}}
\renewcommand{\d}{\bm{d}}
\renewcommand{\i}{\bm{i}}
\newcommand{\bell}{\bm{\ell}}
\newcommand{\m}{\bm{m}}
\newcommand{\p}{\bm{p}}
\newcommand{\q}{\bm{q}}
\renewcommand{\r}{\bm{r}}
\newcommand{\s}{\bm{s}}
\renewcommand{\t}{\bm{t}}

\newcommand{\x}{\bm{x}}
\newcommand{\y}{\bm{y}}
\newcommand{\z}{\bm{z}}

\newcommand{\bdelta}{\bm{\delta}}
\newcommand{\bgamma}{\bm{\gamma}}

\newcommand{\I}{\bm{I}}

\newcommand{\X}{\bm{X}}
\newcommand{\Y}{\bm{Y}}

\newcommand{\PI}{\bm{\Pi}}

\newcommand{\calD}{\mathcal{D}}

\newcommand{\calT}{\mathcal{T}}
\newcommand{\calU}{\mathcal{U}}
\newcommand{\calX}{\mathcal{X}}
\newcommand{\calY}{\mathcal{Y}}

\newcommand{\tOmega}{\tilde{\Omega}}

\newcommand{\ol}[1]{\overline{#1}}

\newcommand{\aPi}{\ol{\Pi}}

\let\OLDthebibliography\thebibliography
\renewcommand\thebibliography[1]{
  \OLDthebibliography{#1}
}

\begin{document}

\title{Query-to-Communication Lifting for $\BPP$}

\author{\setlength\tabcolsep{1.2em}
\begin{tabular}{ccc}
Mika G\"o\"os &
Toniann Pitassi &
Thomas Watson\footnote{Supported by NSF grant CCF-1657377.}\\[0mm]
\slshape\small Harvard University and &
\slshape\small University of Toronto &
\slshape\small University of Memphis \\[-1.1mm]
\slshape\small Simons Institute
\end{tabular}}

\date{\large\today}

\maketitle

\begin{abstract}\noindent
For any $n$-bit boolean function $f$, we show that the randomized communication complexity of the composed function $f\circ g^n$, where $g$ is an index gadget, is characterized by the randomized decision tree complexity of $f$. In particular, this means that many query complexity separations involving randomized models (e.g., classical vs.\ quantum) automatically imply analogous separations in communication complexity.
\end{abstract}

\section{Introduction} \label{sec:intro}

A {\bf\itshape query-to-communication lifting theorem} (a.k.a.\ communication-to-query simulation theorem) translates lower bounds on some type of \emph{query complexity} (a.k.a.~decision tree complexity)~\cite{vereshchagin99relativizability,buhrman02complexity,jukna12boolean} of a boolean function $f$ into lower bounds on a corresponding type of \emph{communication complexity}~\cite{kushilevitz97communication,jukna12boolean,rao17communication} of a two-party version of $f$. See \autoref{tab:lifting} for a list of several known results in this vein. In this work, we show a lifting theorem for bounded-error randomized (i.e., $\BPP$-type) query/com\-munication complexity. Such a theorem had been conjectured by~\cite{anshu16separations,bendavid16randomized,chattopadhyay17composition,wu17raz} and (ad nauseam) by the current authors.

\subsection{Our result}
For a function $f\colon\{0,1\}^n\to\{0,1\}$ (called the \emph{outer function}) and a two-party function $g\colon \calX\times\calY\to\{0,1\}$ (called the \emph{gadget}), their composition $f\circ g^n\colon\calX^n\times\calY^n\to\{0,1\}$ is defined by
\[
(f\circ g^n)(x,y)~\coloneqq~f(g(x_1,y_1),\ldots,g(x_n,y_n)).
\]
Here, Alice holds $x\in \calX^n$ and Bob holds $y\in\calY^n$. Our result is proved for the popular \emph{index} gadget $\Ind_m\colon[m]\times\{0,1\}^m \to\{0,1\}$ mapping $(x,y)\mapsto y_x$. We use $\BPPdt$ and $\BPPcc$ to denote the usual bounded-error randomized query and communication complexities. That is, $\BPPdt(f)$ is the minimum cost of a randomized decision tree (distribution over deterministic decision trees) which, on each input $z$, outputs $f(z)$ with probability at least $2/3$, where the cost is the maximum number of queries over all inputs and outcomes of the randomness; $\BPPcc(F)$ is defined similarly but with communication protocols instead of decision trees.

\begin{theorem}[Lifting for $\BPP$] \label{thm:main}
Let $m=m(n)\coloneqq n^{256}$. For every $f\colon\{0,1\}^n\to\{0,1\}$,
\[
\BPPcc(f\circ \Ind_m^n)~=~\BPPdt(f)\cdot\Theta(\log n).
\]
\end{theorem}
\vfill

\subsection{What does it mean?}

The upshot of our lifting theorem is that it \emph{automates} the task of proving randomized communication lower bounds: we only need to show a problem-specific query lower bound for $f$ (which is often relatively simple), and then invoke the general-purpose lifting theorem to completely characterize the randomized communication complexity of $f\circ \Ind_m^n$.

\paragraph{Separation results.}
The lifting theorem is especially useful for constructing examples of two-party functions that have large randomized communication complexity, but low complexity in some other communication model. For example, one of the main results of Anshu et al.~\cite{anshu16separations} is a nearly $2.5$-th power separation between randomized and quantum ($\BQPcc$) communication complexities for a total function $F$:
\begin{equation} \label{eq:sep}
\BPPcc(F)~\geq~\BQPcc(F)^{2.5-o(1)}.
\end{equation}
Previously, a quadratic separation was known (witnessed by set-disjointness). The construction of~$F$ (and its ad hoc analysis) in \cite{anshu16separations} was closely modeled after an analogous query complexity separation, $\BPPdt(f)\geq\BQPdt(f)^{2.5-o(1)}$, shown earlier by~\cite{aaronson16separations}. Our lifting theorem can reproduce the separation~\eqref{eq:sep} by simply taking $F\coloneqq f\circ \Ind_m^n$ and using the query result of~\cite{aaronson16separations} as a black-box. Here we only note that $\BQPcc(F)$ is at most a logarithmic factor larger than $\BQPdt(f)$, since a protocol can always efficiently simulate a decision tree.

In a similar fashion, we can unify (and in some cases simplify) several other existing results in communication complexity~\cite{raz99exponential,goos15randomized,anshu16separations,babichenko17communication}, including separations between $\BPPcc$ and the log of the partition number; see \autoref{sec:applications} for details. At the time of the writing, we are not aware of any \emph{new} applications implied by our lifting theorem.

\paragraph{Gadget size.}
A drawback with our lifting theorem is that it assumes gadget size $m=\poly(n)$, which limits its applicability. For example, we are not able to reproduce tight randomized lower bounds for important functions such as set-disjointness~\cite{kalyanasundaram92probabilistic,razborov92distributional,bar-yossef04information} or gap-Hamming~\cite{chakrabarti12optimal,sherstov12communication,vidick13concentration}. It remains an open problem to prove a lifting theorem for~$m=O(1)$ even for the models studied in~\cite{goos16rectangles,kothari17approximating}.

\begin{table}
\centering
\renewcommand{\arraystretch}{1.1}
\setlength{\tabcolsep}{3mm}
\begin{tabular}{lllp{5cm}<{\raggedright\small}}
\toprule[.5mm]
\bf Class & \bf Query model & \bf Communication model & \bf References \\
\midrule
$\P$ & deterministic & deterministic & 
\cite{raz99separation,goos15deterministic,rezende16limited,hatami16structure,wu17raz,chattopadhyay17composition} \\
$\NP$ & nondeterministic & nondeterministic & 
\cite{goos16rectangles,goos15lower} \\
\emph{many} & polynomial degree & rank & 
\cite{shi09quantum,sherstov11pattern,razborov10sign,robere16exponential} \\
\emph{many} & conical junta degree & nonnegative rank & 
\cite{goos16rectangles,kothari17approximating} \\
$\P^\NP$ & decision list & rectangle overlay & 
\cite{goos17query} \\
\midrule
& Sherali--Adams & LP extension complexity & 
\cite{chan16approximate,kothari17approximating} \\
& sum-of-squares & SDP extension complexity &
\cite{lee15lower} \\
\bottomrule[.5mm]
\end{tabular}
\caption{Query-to-communication lifting theorems. The first five are formulated in the language of boolean functions (as in this paper); the last two are formulated in the language of combinatorial optimization.}
\label{tab:lifting}
\end{table}

\section{Reformulation}

Our lifting theorem holds for all $f$, even if $f$ is a partial function or a general relation (search problem). Thus the theorem is not \emph{really} about the outer function at all; it is about the obfuscating ability of the index gadget $\Ind_m$ to hide information about the input bits of $f$. To focus on what is essential, let us reformulate the lifting theorem in a more abstract way that makes no reference to $f$.

\subsection{Slices}
Write $G\coloneqq g^n$ for $g\coloneqq \Ind_m$. We view $G$'s input domain $[m]^n\times(\{0,1\}^m)^n$ as being partitioned into \emph{slices} $G^{-1}(z)=\{(x,y):G(x,y)=z\}$, one for each $z\in\{0,1\}^n$; see (a) below. We will eventually consider \emph{randomized} protocols, but suppose for simplicity that we are given a \emph{deterministic} protocol $\Pi$ of communication cost $|\Pi|$. The most basic fact about $\Pi$ is that it induces a partition of the input domain into at most $2^{|\Pi|}$ rectangles (sets of the form $X\times Y$ where $X\subseteq[m]^n$, $Y\subseteq(\{0,1\}^m)^n$); see (b) below. The rectangles are in 1-to-1 correspondence with the leaves of the protocol tree, which are in 1-to-1 correspondence with the protocol's \emph{transcripts} (root-to-leaf paths; each path is a concatenation of messages). Fixing some $z\in\{0,1\}^n$, we are interested in the distribution over transcripts that is generated when $\Pi$ is run on a uniform random input from the slice $G^{-1}(z)$; see (c) below.
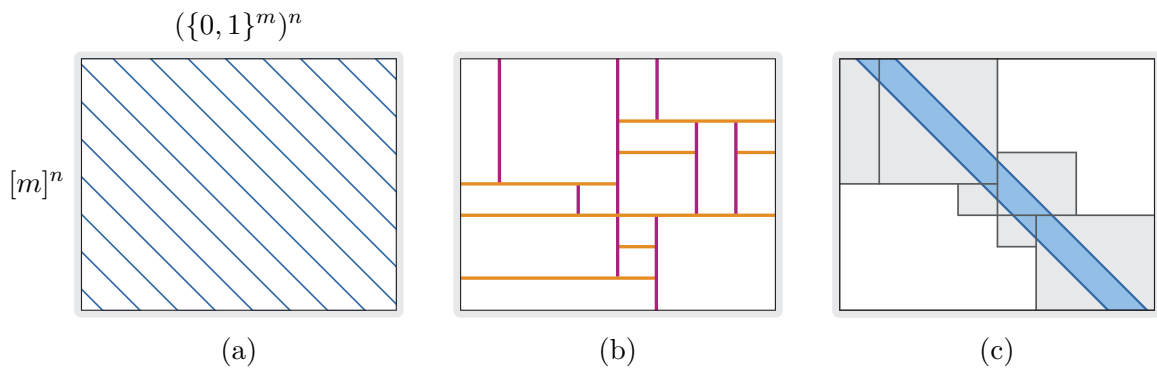
\begin{center}
\begin{lpic}[b(-1mm),l(-3mm),r(-3mm),t(-1mm)]{slices(.42)}
\lbl[c]{6,60;$[m]^n$}
\lbl[c]{70,110;$(\{0,1\}^m)^n$}
\normalsize
\lbl[c]{70,7;(a)}
\lbl[c]{190,7;(b)}
\lbl[c]{310,7;(c)}
\end{lpic}
\end{center}

\subsection{The reformulation}
We devise a \emph{randomized} decision tree that on input $z$ outputs a random transcript distributed close (in total variation distance) to that generated by $\Pi$ on input $(\x,\y)\sim G^{-1}(z)$. (We always use boldface letters for random variables.)

\begin{theorem} \label{thm:goal}
Let $\Pi$ be a deterministic protocol with inputs from the domain of $G=g^n$. There is a randomized decision tree of cost $O(|\Pi|/\log n)$ that on input $z\in\{0,1\}^n$ samples a random transcript (or outputs $\bot$ for failure) such that the following two distributions are $o(1)$-close:
\begin{align*}
\t_z~&\coloneqq~\text{output distribution of the randomized decision tree on input $z$}, \\[-1mm]
\t'_z~&\coloneqq~\text{transcript generated by $\Pi$ when run on a random input $(\x,\y)\sim G^{-1}(z)$}.
\end{align*}
Moreover, the simulation has ``one-sided error'': $\supp(\t_z)\subseteq\supp(\t'_z)\cup \{\bot\}$ for every $z$.
\end{theorem}

The lifting theorem (\autoref{thm:main}) follows as a simple consequence of the above reformulation. For the easy direction (``$\leq$''), any randomized decision tree for $f$ making $c$ queries can be converted into a randomized protocol for $f \circ g^n$ communicating $c \cdot O(\log n)$ bits, where the $O(\log n)$ factor is the deterministic communication complexity of the gadget. For the nontrivial direction (``$\geq$''), suppose we have a randomized protocol $\PI$ (viewed as a probability distribution over deterministic protocols) that computes $f \circ g^n$ (with error $\leq 1/3$, say) and each $\Pi\sim\PI$ communicates at most $|\Pi|\leq c$ bits. We convert this into a randomized decision tree for $f$ of query cost $O(c/\log n)$ as follows.

\medskip\noindent
{\slshape On input $z$:}
\begin{enumerate}[topsep=1mm]
\item Pick a deterministic~$\Pi\sim\PI$ (using random coins of the decision tree).
\item Run the randomized decision tree for $\Pi$ from \autoref{thm:goal} that samples a transcript $t\sim\t_z(\Pi)$.
\item Output the value of the leaf reached in $t$.
\end{enumerate}
\medskip
The resulting decision tree has bounded error on input $z$:
\begin{align*}
\Pr[\,\text{output of decision tree}\, \neq f(z)]
~&=~\Ex_{\Pi\sim \PI}\bigl[\Pr_{t\sim \t_z(\Pi)}[\,\text{value of leaf in $t$}\, \neq f(z)]\bigr]\\
~&=~\Ex_{\Pi\sim \PI}\bigl[\Pr_{t\sim \t'_z(\Pi)}[\,\text{value of leaf in $t$}\, \neq f(z)]\pm o(1)\bigr]\\
~&=~\Ex_{\Pi\sim \PI}\bigl[\Pr_{(\x,\y)\sim G^{-1}(z)}[\Pi(\x,\y) \neq f(z)]\bigr]\pm o(1)\\
~&=~\Ex_{(x,y)\sim G^{-1}(z)}\bigl[\Pr_{\PI}[\PI(x,y) \neq f(z)]\bigr]\pm o(1)\\
~&\leq~\Ex_{(x,y)\sim G^{-1}(z)}[1/3] \pm o(1)\\
~&\leq~1/3+o(1).
\end{align*}

\subsection{Extensions}
The correctness of our simulation hinged on the property of $\BPP$-type algorithms that the \emph{mixture of correct output distributions is correct}. In fact, the ``moreover'' part in \autoref{thm:goal} allows us to get a lifting theorem for \emph{one-sided error} ($\RP$-type) and \emph{zero-sided error} ($\ZPP$-type) query/communication complexity: if the randomized protocol $\PI$ on every input $(x,y)\in G^{-1}(z)$ outputs values in $\{f(z),\bot\}$, so does our decision tree simulation on input $z$. Funnily enough, it was previously known that the existence of a query-to-communication lifting theorem for $\ZPP$ (for index gadget) implies the existence of a lifting theorem for $\BPP$ in a black-box fashion~\cite{bendavid16randomized}. We also mention that \autoref{thm:goal} in fact holds with $1/\!\poly(n)$-closeness (instead of $o(1)$) for an arbitrarily high degree polynomial, provided $m$ is chosen to be a correspondingly high enough degree polynomial in $n$.

\section{Simulation}
We now prove \autoref{thm:goal}. Fix a deterministic protocol $\Pi$ henceforth. We start with a high-level sketch of the simulation, and then fill in the details.

\subsection{Executive summary}

\begin{wrapfigure}[9]{r}{6.5cm}
\begin{lpic}[t(-11mm),l(2mm),r(-10mm)]{marginals(.48)}
\large
\lbl[c]{132,60;$X$}
\lbl[c]{70,7;$Y$}
\lbl[c]{65,59;\rotatebox{-45}{$G^{-1}(z)\cap X\!\times\!Y$}}
\end{lpic}
\end{wrapfigure}

The randomized decision tree will generate a random transcript of $\Pi$ by taking a random walk down the protocol tree of $\Pi$, guided by occasional queries to the bits of $z$. The design of our random walk is dictated by one (and only one) property of the slice sets~$G^{-1}(z)$:

\begin{itemize}[leftmargin=\parindent]
\item[]
\emph{\bf Uniform marginals lemma (informal):}\\
For every $z\in\{0,1\}^n$ and every rectangle $X \times Y$ where $X$ is ``dense'' and~$Y$ is ``large'', the uniform distribution on $G^{-1}(z) \cap X \times Y$ has both of its marginal distributions close to uniform on $X$ and $Y$, respectively.
\end{itemize}

\noindent
This immediately suggests a way to \emph{begin} the randomized simulation. Each node of $\Pi$'s protocol tree is associated with a rectangle $X \times Y$ of all inputs that reach that node. We start at the root where, initially, $X\times Y = [m]^n\times(\{0,1\}^m)^n$. Suppose Alice communicates the first bit~$b \in \{0,1\}$. This induces a partition $X = X^0 \cup X^1$ where $X^b$ consists of those inputs where Alice sends~$b$. When $\Pi$ is run on a random input $(\x,\y)\sim G^{-1}(z)$, the above lemma states that $\x$ is close to uniform on~$X$ and hence the branch $X^b$ is taken with probability roughly~$|X^b|/|X|$. Our idea for a simulation is this: we pretend that $\x \sim X$ is perfectly uniform so that our simulation takes the branch~$X^b$ with probability exactly $|X^b|/|X|$. It follows that the first bit sent in the two scenarios ($\t_z$ and $\t'_z$) is distributed close to each other. We can continue the simulation in the same manner, updating~$X \leftarrow X^b$ (and similarly $Y \leftarrow Y^b$ when Bob speaks), as long as $X \times Y$ remains ``$\text{dense}\times\text{large}$''.

\paragraph{Largeness.}
A convenient property of the index gadget is that Bob's $nm$-bit input is much longer than Alice's $n\log m$-bit input. Consequently, the simulation will not need to go out of its way to maintain the ``largeness'' of Bob's set $Y$---we will argue that it naturally remains ``large'' enough with high probability throughout the simulation.

\paragraph{Density.}
The interesting case is when Alice's set $X$ ceases to be ``dense''. Our idea is to promptly restore ``density'' by computing a \emph{density-restoring} partition $X = \bigcup_i X^i$ with the property that each $X^i$ is fixed on some subset of blocks $I_i \subseteq [n]$ (which ``caused'' a density violation), and such that $X^i$ is again ``dense'' on the remaining blocks $[n]\smallsetminus I_i$. Moreover, $|I_i|$ will typically be bounded in terms of the number of bits communicated so far.

After Alice has partitioned $X=\bigcup_i X^i$ we will follow the branch $X^i$ (updating $X\leftarrow X^i$) with probability $|X^i|/|X|$; this random choice is justified by the uniform marginals lemma, since it imitates what would happen on a uniform random input from $G^{-1}(z)$. Since we made Alice's pointers~$X^i_{I_i}$ fixed, say, to value $\alpha\in[m]^{I_i}$, we need to fix the corresponding pointed-to bits on Bob's side so as to make the output of the gadgets $g^n(X^i,Y)$ consistent with $z$ on the fixed coordinates. At this point, our decision tree queries all the bits $z_{I_i}\in\{0,1\}^{I_i}$ and we argue that we can indeed typically restrict Bob's set to some still-``large'' $Y^i\subseteq Y$ to ensure $g^{I_i}(X^i_{I_i}\times Y^i_{I_i})=\{z_{I_i}\}$. Now that we have recovered ``density'' on the unfixed blocks, we may continue the simulation as before (relativized to unfixed blocks).

\subsection{Tools}
Let us make the notions of ``dense'' and ``large'' precise. Let $\Hmin(\x) \coloneqq \min_x\log(1/\Pr[\x=x])$ denote the usual min-entropy of a random variable $\x$. Supposing $\x$ is distributed over a set $X$, we define the \emph{deficiency} of $\x$ as the nonnegative quantity $\Dmin(\x)\coloneqq \log|X|- \Hmin(\x)$. A basic property, which we use freely and repeatedly throughout the proof, is that marginalizing $\x$ to some coordinates (assuming $X$ is a product set) cannot increase the deficiency. For a set $X$ we use the boldface $\X$ to denote a random variable uniformly distributed on $X$. 

\begin{definition}[Blockwise-density~\cite{goos16rectangles}] \label{def:dense}
A random variable $\x\in[m]^J$ (where $J$ is some index set) is called \emph{$\delta$-dense} if for every nonempty $I\subseteq J$, the blocks $\x_I$ have min-entropy rate at least $\delta$, that is, $\Hmin(\x_I) \geq \delta\cdot |I|\log m$. (Note that $\x_I$ is marginally distributed over $[m]^I$.)
\end{definition}

\begin{lemma}[Uniform marginals; simple version] \label{lem:simple}
Suppose $\X$ is $0.9$-dense and $\Dmin(\Y)\leq n^3$. Then for any $z\in\{0,1\}^n$, the uniform distribution on $G^{-1}(z) \cap X \times Y$ (which is nonempty) has both of its marginal distributions $1/n^2$-close to uniform on $X$ and $Y$, respectively.
\end{lemma}

We postpone the proof of the lemma to \autoref{sec:pseudorandomness}, and instead concentrate here on the simulation itself---its correctness will mostly rely on this lemma. Actually, we need a slightly more general-looking statement that we can easily apply when some blocks in $X$ have become fixed during the simulation. To this end, we introduce terminology for such rectangles $X\times Y$. Note that \autoref{lem:general} below specializes to \autoref{lem:simple} by taking $\rho=*^n$.

\begin{definition}[Structured rectangles]
For a partial assignment $\rho\in\{0,1,*\}^n$, define its \emph{free} positions as $\free\rho\coloneqq\rho^{-1}(*)\subseteq[n]$, and its \emph{fixed} positions as $\fix\rho\coloneqq [n]\smallsetminus\free\rho$. A rectangle $X\times Y$ is called \emph{$\rho$-structured} if $\X_{\free\rho}$ is $0.9$-dense, $\X_{\fix\rho}$ is fixed, and each output in $G(X\times Y)$ is consistent with $\rho$.
\end{definition}

\vspace{-3mm}
\begin{restatable}[Uniform marginals; general version]{lemma}{unifmarg} \label{lem:general}
Suppose $X\times Y$ is $\rho$-structured and $\Dmin(\Y)\leq n^3$. Then for any $z\in\{0,1\}^n$ consistent with $\rho$, the uniform distribution on $G^{-1}(z) \cap X \times Y$ (which is nonempty) has both of its marginal distributions $1/n^2$-close to uniform on $X$ and $Y$, respectively.
\end{restatable}

\begin{center}
\begin{lpic}[b(5mm),t(1mm),r(3mm)]{structured(.25)}

\normalsize
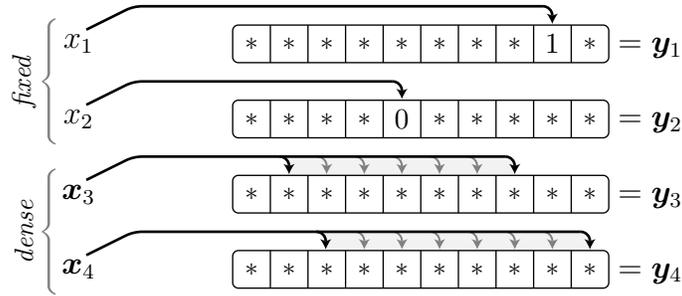
\lbl[c]{190,-10;\slshape Illustration of $\x\sim X$ and $\y\sim Y$ where $X\times Y$ is $\rho$-structured for $\rho\coloneqq10\!*\!*$}

\large
\lbl[c]{38,150;$x_1$}
\lbl[c]{38,110;$x_2$}
\lbl[c]{38,70;$\x_3$}
\lbl[c]{38,30;$\x_4$}

\small
\lbl[c]{9,130;\rotatebox{90}{\itshape fixed}}
\lbl[c]{9,50;\rotatebox{90}{\itshape dense}}

\large
\lbl[l]{325,148;$=\y_1$}
\lbl[l]{325,108;$=\y_2$}
\lbl[l]{325,68;$=\y_3$}
\lbl[l]{325,28;$=\y_4$}

\normalsize
\lbl[c]{290,150;$1$}
\lbl[c]{210,110;$0$}

\small
\lbl[c]{130,150;$*$}
\lbl[c]{150,150;$*$}
\lbl[c]{170,150;$*$}
\lbl[c]{190,150;$*$}
\lbl[c]{210,150;$*$}
\lbl[c]{230,150;$*$}
\lbl[c]{250,150;$*$}
\lbl[c]{270,150;$*$}
\lbl[c]{310,150;$*$}

\lbl[c]{130,110;$*$}
\lbl[c]{150,110;$*$}
\lbl[c]{170,110;$*$}
\lbl[c]{190,110;$*$}
\lbl[c]{230,110;$*$}
\lbl[c]{250,110;$*$}
\lbl[c]{270,110;$*$}
\lbl[c]{290,110;$*$}
\lbl[c]{310,110;$*$}

\lbl[c]{130,70;$*$}
\lbl[c]{150,70;$*$}
\lbl[c]{170,70;$*$}
\lbl[c]{190,70;$*$}
\lbl[c]{210,70;$*$}
\lbl[c]{230,70;$*$}
\lbl[c]{250,70;$*$}
\lbl[c]{270,70;$*$}
\lbl[c]{290,70;$*$}
\lbl[c]{310,70;$*$}

\lbl[c]{130,30;$*$}
\lbl[c]{150,30;$*$}
\lbl[c]{170,30;$*$}
\lbl[c]{190,30;$*$}
\lbl[c]{210,30;$*$}
\lbl[c]{230,30;$*$}
\lbl[c]{250,30;$*$}
\lbl[c]{270,30;$*$}
\lbl[c]{290,30;$*$}
\lbl[c]{310,30;$*$}

\end{lpic}
\end{center}

\subsection{Density-restoring partition} \label{sec:density-restoring}

Fix some set $X\subseteq[m]^J$. (In our application, $J\subseteq [n]$ will correspond to the set of free blocks during the simulation.) We describe a procedure that takes $X$ and outputs a \emph{density-restoring} partition $X=\bigcup_i X^i$ such that each $\X^i$ is fixed on some subset of blocks $I_i\subseteq J$ and $0.9$-dense on $J\smallsetminus I_i$. The procedure associates a \emph{label} of the form ``$x_{I_i}=\alpha_i$'' with each part~$X_i$, recording which blocks we fixed and to what value. If $\X$ is already $0.9$-dense, the procedure outputs just one part: $X$ itself.

\medskip\noindent
{\slshape While $X$ is nonempty:}
\begin{enumerate}[topsep=1mm]
\item Let $I\subseteq J$ be a \emph{maximal} subset (possibly $I=\emptyset$) such that $\X_I$ has min-entropy rate $<0.9$, and let $\alpha\in[m]^I$ be an outcome witnessing this: $\Pr[\X_I=\alpha]>m^{-0.9|I|}$.
\item Output part $X^{(x_I=\alpha)}\coloneqq\{x \in X: x_I=\alpha\}$ with label ``$x_I=\alpha$''.
\item Update $X\leftarrow X\smallsetminus X^{(x_I=\alpha)}$.
\end{enumerate}

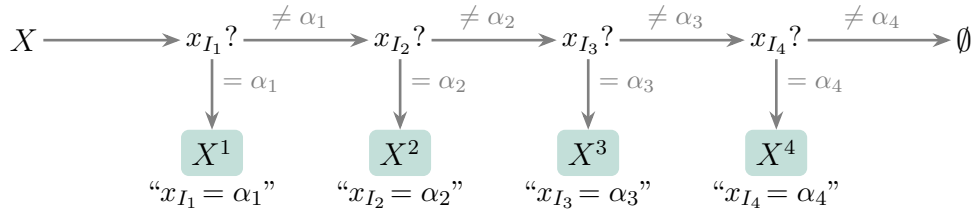
\begin{center}
\begin{tikzpicture}[xscale=2.5,yscale=1.5]
\tikzset{
	inner sep=0,outer sep=2,
	leaf/.style={color=black,fill=white!80!PineGreen,rectangle,
	rounded corners=3pt,minimum height=4mm,minimum width=4mm,inner sep=4}}

\large
\node (start) at (-1,0) {$X$};
\node (e) at (4,0) {$\emptyset$};

\normalsize
\node (a) at (0,0) {$x_{I_1}$?};
\node (b) at (1,0) {$x_{I_2}$?};
\node (c) at (2,0) {$x_{I_3}$?};
\node (d) at (3,0) {$x_{I_4}$?};

\large
\tikzset{
  inner sep=0,outer sep=0}
	
\node[leaf] (o1) at (0,-1) {$X^1$};
\node[leaf] (o2) at (1,-1) {$X^2$};
\node[leaf] (o3) at (2,-1) {$X^3$};
\node[leaf] (o4) at (3,-1) {$X^4$};

\normalsize
\node (l) at (0,-1.37) {``$x_{I_1}\!=\alpha_1$''};
\node (l) at (1,-1.37) {``$x_{I_2}\!=\alpha_2$''};
\node (l) at (2,-1.37) {``$x_{I_3}\!=\alpha_3$''};
\node (l) at (3,-1.37) {``$x_{I_4}\!=\alpha_4$''};

\small
\draw[-{Stealth[length=7pt]},line width=1pt,pos=0.35,inner sep=4,color=gray]
(start) edge (a)
(a) edge node[above,pos=0.45] {$\neq \alpha_1$} (b)
(b) edge node[above,pos=0.45] {$\neq \alpha_2$} (c) 
(c) edge node[above,pos=0.45] {$\neq \alpha_3$} (d)
(d) edge node[above,pos=0.45] {$\neq \alpha_4$} (e)
(a) edge node[right] {$= \alpha_1$} (o1) 
(b) edge node[right] {$= \alpha_2$} (o2)
(c) edge node[right] {$= \alpha_3$} (o3)
(d) edge node[right] {$= \alpha_4$} (o4);
\end{tikzpicture}
\end{center}

We collect below the key properties of the partition $X=\bigcup_i X^i$ output by the procedure. Firstly, the partition indeed restores blockwise-density for the unfixed blocks. Secondly, the deficiency (relative to unfixed blocks) typically decreases proportional to the number of blocks we fixed.
\begin{lemma} \label{lem:density-restoring}
Each $X^i$ (labeled ``$x_{I_i}=\alpha_i$'') in the density-restoring partition satisfies the following.\nopagebreak
\begin{itemize}[topsep=2mm,leftmargin=28mm,itemsep=.5mm]
\item[\slshape (Density):~]
$\X^i_{J\smallsetminus I_i}$ is $0.9$-dense
\item[\slshape (Deficiency):~]
$\Dmin(\X^i_{J\smallsetminus I_i}) \leq \Dmin(\X) - 0.1|I_i|\log m + \delta_i$~~where~~$\delta_i \coloneqq \log (|X|/|\cup_{j\geq i} X^j|)$
\end{itemize}
\end{lemma}
\begin{proof}
Write $X^{\geqslant i}\coloneqq\bigcup_{j\geq i}X^j$ so that $\X^i = (\X^{\geqslant i}\mid \X^{\geqslant i}_{I_i}=\alpha_i)$. Suppose for contradiction that some part~$\X^i$ was not $0.9$-dense on $J\smallsetminus I_i$. Then there is some nonempty $K\subseteq J\smallsetminus I_i$ and an outcome $\beta\in[m]^K$ violating the min-entropy condition: $\Pr[\X^i_K=\beta]>m^{-0.9|K|}$. But this contradicts the maximality of~$I_i$ since the larger set $I_i\cup K$ now violates the min-entropy condition for $\X^{\geqslant i}$:
\[
\Pr[\X^{\geqslant i}_{I_i\cup K}=\alpha_i\beta]
~=~\Pr[\X^{\geqslant i}_{I_i}=\alpha_i]\cdot\Pr[\X^i_K=\beta]
~>~m^{-0.9|I_i|}\cdot m^{-0.9|K|}
~=~m^{-0.9|I_i\cup K|}.
\]
This proves the first part. The second part is a straightforward calculation (intuitively, going from $X$ to $X^{\geqslant i}$ causes a $\delta_i$ increase in deficiency, going from $X^{\geqslant i}$ to $X^i$ causes a $\le 0.9|I_i|\log m$ increase, and restricting from $J$ to $J\smallsetminus I_i$ causes a $|I_i|\log m$ decrease):
\begin{align*}
\Dmin(\X^i_{J\smallsetminus I_i})~&=~|J\smallsetminus I_i|\log m-\log|X^i|\\
&\le~\bigl(|J|\log m-|I_i|\log m\bigr)-\log\bigl(|X^{\geqslant i}|\cdot 2^{-0.9|I_i|\log m}\bigr)\\
&=~\bigl(|J|\log m-\log|X|\bigr)-0.1|I_i|\log m+\log\bigl(|X|/|X^{\geqslant i}|\bigr)\\
&=~\Dmin(\X)-0.1|I_i|\log m+\delta_i.\qedhere
\end{align*}
\end{proof}

\subsection{The simulation}

To describe our simulation in a convenient language, we modify the deterministic protocol $\Pi$ into a \emph{refined} deterministic protocol $\aPi$; see \autoref{fig:protocol}. Namely, we insert two new rounds of communication whose sole purpose is to restore density for Alice's free blocks by fixing some other blocks and Bob's corresponding bits. In short, we maintain the rectangle $X\times Y$ as $\rho$-structured for some $\rho$. Each communication round of $\Pi$ is thus replaced with a whole \emph{iteration} in $\aPi$. The new communication rounds do not affect the input/output behavior of the original protocol: any transcript of $\aPi$ can be projected back to a transcript of~$\Pi$ (by ignoring messages sent on lines 14, 16). One way to think about $\aPi$ is that it induces a partition of the communication matrix that is a \emph{refinement} of the one $\Pi$ induces. Therefore, for the purpose of proving \autoref{thm:goal}, we can concentrate on simulating $\aPi$ in place of $\Pi$. The randomized decision tree becomes simple to describe relative to $\aPi$; see \autoref{fig:decision-tree}.

Next, we proceed to show that our randomized decision tree is \emph{(1)} correct: on input $z$ it samples a transcript distributed close to that of $\aPi$ when run on $(\x,\y)\sim G^{-1}(z)$, and \emph{(2)} efficient: the number of queries it makes is bounded in terms of $|\Pi|$ (the number of iterations in $\aPi$).

\begin{figure}[p]
\newcommand{\myline}{
\vspace{-1mm}\hspace{-7mm}
\begin{tikzpicture}
\draw [thick,dashed] (0,0) -- (15.1,0);
\end{tikzpicture}\vspace{0.5mm}}

\newcommand{\myprint}[1]{\hspace{-3.55mm}\raisebox{.4mm}{\scriptsize#1}\hspace{1.35mm}}
\newcommand{\mymark}{\myprint{$\vartriangleright$}}
\newcommand{\mymarkalt}{\myprint{$\blacktriangleright$}}

\algblockdefx[While]{While}{EndWhile}{{\bf while}~}{{\bf end while}}

\begin{mdframed}
{\bf\slshape\large Refined protocol\hspace{1pt} $\aPi$ on input $(x,y)$:}
\vspace{2.5mm}
\begin{algorithmic}[1]
\State initialize:~~ $v = \text{root of $\Pi$}$,~~ $X\times Y=[m]^n\times(\{0,1\}^m)^n$,~~ {\color{Maroon}$\rho = *^n$}
\While {{$v$ is not a leaf}}~~
{\color{Maroon}
{\bf[}\,invariant: $X\times Y$ is $\rho$-structured\,{\bf]}
}
\State let $v_0$, $v_1$ be the children of $v$
\If {{Bob sends a bit at $v$}}
\State let $Y=Y^0\cup Y^1$ be the partition according to Bob's function at $v$
\State let $b$ be such that $y\in Y^b$
\State \mymark\underline{Bob sends $b$} and we update $Y\leftarrow Y^b$, $v\leftarrow v_b$
\Else {{ Alice sends a bit at $v$}}
\State let $X=X^0\cup X^1$ be the partition according to Alice's function at $v$
\State let $b$ be such that $x\in X^b$
\State \mymark\underline{Alice sends $b$} and we update $X\leftarrow X^b$, $v\leftarrow v_b$
{\color{Maroon}
\Statex \myline
\State let $X=\bigcup_i X^i$ be such that \smash{$X_{\free\rho}=\bigcup_i X^i_{\free\rho}$} is a density-restoring partition
\State let $i$ be such that $x\in X^i$ and suppose \smash{$X^i_{\free\rho}$} is labeled ``$x_I=\alpha$'', $I\subseteq\free\rho$
\State \mymark\underline{Alice sends $i$} and we update $X\leftarrow X^i$
\State let $s = g^I(\alpha,y_I)\in\{0,1\}^I$
\State \mymarkalt\underline{Bob sends $s$} and we update $Y\leftarrow \{y'\in Y: g^I(\alpha,y'_I)=s\}$, $\rho_I\leftarrow s$
\Statex \myline
}
\EndIf
\EndWhile
\State output the value of the leaf $v$
\end{algorithmic}
\end{mdframed}

\caption{The refined (deterministic) protocol $\aPi$. The protocol explicitly keeps track of a rectangle $X\times Y$ consisting of all inputs that reach the current node (i.e., produce the same transcript so far). The original protocol $\Pi$ can be recovered by simply ignoring lines~12--16 and text in {\color{Maroon}\bf red}. The purpose of lines 12--16 is to maintain the invariant; they do not affect the input/output behavior.}\label{fig:protocol}
\end{figure}

\begin{figure}[p]
\begin{mdframed}
{\bf\slshape\large Randomized decision tree on input $z$:}
\vspace{3mm}

To generate a transcript of $\aPi$ we take a random walk down $\aPi$'s protocol tree, guided by queries to the bits of~$z$. The following defines the distribution of messages to send at each underlined line.
\begin{description}[itemsep=1mm,leftmargin=7mm]
\item[Lines marked `$\vartriangleright$':] We simulate an iteration of the protocol $\aPi$ pretending that $\x\sim X$ and $\y\sim Y$ are uniformly distributed over their domains. Namely, in line 7, we send~$b$ with probability~$|Y^b|/|Y|$; in line 11, we send $b$ with probability $|X^b|/|X|$; in line 14 (after having updated $X\leftarrow X^b$), we send~$i$ with probability $|X^i|/|X|$.
\item[Line marked `$\blacktriangleright$':] Here we query $z_I$ and send \emph{deterministically} the message~$s=z_I$; except if this message is impossible to send (because $z_I\notin g^I(\alpha,Y_I)$), we output $\bot$ and halt the simulation with failure.
\end{description}
\end{mdframed}

\caption{The randomized decision tree with query access to $z$. Its goal is to generate a random transcript of $\aPi$ that is $o(1)$-close to the transcript generated by $\aPi$ on a random input $(\x,\y)\sim G^{-1}(z)$.}
\label{fig:decision-tree}%
\end{figure}

\subsection{Correctness: Transcript distribution}
We show that for every $z\in\{0,1\}^n$ the following distributions are $o(1)$-close:
\begin{align*}
\t~&\coloneqq~\text{transcript generated by our simulation of\, $\aPi$ with query access to $z$},\\
\t'~&\coloneqq~\text{transcript generated by $\aPi$ when run on a random input from $G^{-1}(z)$}.
\end{align*}

The following is the heart of the argument.

\begin{lemma} \label{lem:iteration}
Consider a node $v$ at the beginning of an iteration in $\aPi$'s protocol tree, such that $z$ is consistent with the associated $\rho$. Suppose $X\times Y$ is the $\rho$-structured rectangle at $v$, and assume that $\Dmin(\Y)\le n^3$. Let $\m$ and $\m'$ denote the messages sent in this iteration under $\t$ and $\t'$ respectively (conditioned on reaching $v$). Then
\begin{itemize}
\item[(i)] $\m$ and $\m'$ are $1/n^2$-close,
\item[(ii)] with probability at least $1-4/n^2$ over $\m$, at least a $2^{-(n\log m+2)}$ fraction of $Y$ is retained.
\end{itemize}
\end{lemma}

Before proving the lemma, let us use it to show that $\t$ and $\t'$ are $o(1)$-close. For this, it suffices to exhibit a coupling such that $\Pr[\t=\t']\ge 1-o(1)$. Our coupling works as follows:

\medskip\noindent
{\slshape Begin at the root, and for each iteration of $\aPi$:}
\begin{enumerate}[topsep=1mm]
\item Sample this iteration's messages $\m$ and $\m'$ according to an optimal coupling.
\item If $\m\ne\m'$, or if $\m$ results in $<2^{-(n\log m+2)}$ fraction of $Y$ being retained (this includes the simulation's failure case), then proceed to sample the rest of $\t$ and $\t'$ independently.
\end{enumerate}
\medskip
It follows by induction on $k$ that after the $k$-th iteration, with probability at least $1-k\cdot 5/n^2$,
\begin{enumerate}
\item[(I)] $\t$ and $\t'$ match so far,
\item[(II)] $\Dmin(\Y)\le k\cdot(n\log m+2)\le n^3$ where $Y$ is Bob's set under $\t$ so far.
\end{enumerate}
This trivially holds for $k=0$. For $k>0$, conditioned on (I) and (II) for iteration $k-1$, the assumptions of \autoref{lem:iteration} are met and hence $\Pr[\m=\m']\ge 1-1/n^2$ and \[\Pr\bigl[\Dmin(\Y)\le (k-1)\cdot(n\log m+2)+(n\log m+2)=k\cdot(n\log m+2)\bigr]~\ge~1-4/n^2.\] By a union bound, with probability $\ge 1-5/n^2$, (I) and (II) continue to hold. Thus, \[\Pr[\text{(I) and (II) hold after the $k$-th iteration}]~\ge~(1-(k-1)\cdot 5/n^2)\cdot (1-5/n^2)~\ge~1-k\cdot 5/n^2.\] Since there are at most $n\log m$ iterations, we indeed always have $k\cdot(n\log m+2)\le n^3$ (in (II)), and in the end we have $\Pr[\t=\t']\ge 1-(n\log m)\cdot 5/n^2\ge 1-o(1)$ and thus $\t$ and $\t'$ are $o(1)$-close.

\begin{proof}[Proof of \autoref{lem:iteration}]
Let $\x\coloneqq\X$ be uniform over $X$, and $\y\coloneqq\Y$ be uniform over $Y$, and $(\x',\y')$ be uniform over $G^{-1}(z)\cap X\times Y$. By \autoref{lem:general}, $\x$ and $\x'$ are $1/n^2$-close, and $\y$ and $\y'$ are $1/n^2$-close.

First assume Bob sends a bit at $v$. Then $\m$ is some deterministic function of $\y$, and $\m'$ is the same deterministic function of $\y'$ (the bit sent on line 7); thus $\m$ and $\m'$ are $1/n^2$-close since $\y$ and $\y'$ are. Also, the second property in the lemma statement trivially holds.

Henceforth assume Alice sends a bit at $v$. Write $\m=\b\i\s$ (jointly distributed with $\x$) and $\m'=\b'\!\i'\!\s'$ (jointly distributed with $(\x',\y')$) as the concatenation of the three messages sent (on lines 11, 14, 16). Then $\b\i\s$ is some deterministic function of $\x$, and $\b'\!\i'\!\s'$ is the same deterministic function of $\x'$ ($\s$ and $\s'$ depend on $z$, which is fixed); thus $\m$ and $\m'$ are $1/n^2$-close since $\x$ and $\x'$ are. A subtlety here is that there may be outcomes of $\b\i$ for which $\s$ is not defined (there is no corresponding child in $\aPi$'s protocol tree, since Bob's set would become empty), in which case our randomized decision tree fails and outputs $\bot$. But such outcomes have $0$ probability under $\b'\!\i'$, so it is still safe to say $\m$ and $\m'$ are $1/n^2$-close, treating $\s$ as $\bot$ if it is undefined.

We turn to verifying the second property. Define $X^{bi}\times Y^{bi}\subseteq X\times Y$ as the rectangle at the end of the iteration if Alice sends $b$ and $i$, and note that $\x\in X^{\b\i}$ and $\x'\in X^{\b'\!\i'}$. There is a coupling of $\y$ and $\y'$ such that $\Pr[\y\ne\y']\le 1/n^2$; we may imagine that $\y$ is jointly distributed with $(\x',\y')$: sample $(\x',\y')$ and then conditioned on the outcome of $\y'$, sample $\y$ according to the coupling. Note that for each $bi$, \[\Pr[\y\in Y^{bi}]~\ge~\Pr[\y\in Y^{bi}\mid\x'\in X^{bi}]\cdot\Pr[\x'\in X^{bi}]~\ge~\Pr[\y=\y'\mid\x'\in X^{bi}]\cdot\Pr[\x'\in X^{bi}]\] (since $\x'\in X^{bi}$ implies $\y'\in Y^{bi}$), and so
\begin{equation} \label{eq:bob1}
\Pr_{bi\sim\b'\!\i'}\bigl[\Pr[\y\in Y^{bi}]<\Pr[\x'\in X^{bi}]/2\bigr]~\le~\Pr_{bi\sim\b'\!\i'}\bigl[\Pr[\y\ne\y'\mid\x'\in X^{bi}]\ge 1/2\bigr]~\le~2/n^2.
\end{equation}
It is also straightforward to check that
\begin{equation} \label{eq:bob2}
\Pr_{bi\sim\b'\!\i'}\bigl[\Pr[\x'\in X^{bi}]<\Pr[\x\in X^{bi}]/2\bigr]~\le~1/n^2.
\end{equation}
Since trivially $\Pr[\x\in X^{bi}]\ge 1/|X|\ge 2^{-n\log m}$, combining \eqref{eq:bob1} and \eqref{eq:bob2} we have
\begin{align*}
&\Pr_{bi\sim\b\i}\bigl[\Pr[\y\in Y^{bi}]<2^{-(n\log m+2)}\bigr]\\
\le~{}&\Pr_{bi\sim\b'\!\i'}\bigl[\Pr[\y\in Y^{bi}]<2^{-(n\log m+2)}\bigr]+1/n^2\\
\le~{}&\Pr_{bi\sim\b'\!\i'}\bigl[\Pr[\y\in Y^{bi}]<\Pr[\x\in X^{bi}]/4\bigr]+1/n^2\\
\le~{}&\Pr_{bi\sim\b'\!\i'}\Bigl[\Pr[\y\in Y^{bi}]<\Pr[\x'\in X^{bi}]/2\,\text{ or }\,\Pr[\x'\in X^{bi}]<\Pr[\x\in X^{bi}]/2\Bigr]+1/n^2\\
\le~{}&2/n^2+1/n^2+1/n^2.\qedhere
\end{align*}
\end{proof}

\paragraph{``One-sided error''.}
One more detail to iron out is the ``moreover'' part in the statement of \autoref{thm:goal}. The simulation we described does not quite satisfy this condition, but this is simple to fix: instead of halting with failure only when $Y$ becomes empty, we actually halt with failure when $\Dmin(\Y)>n^3$. This does not affect the correctness or efficiency analysis at all, but it ensures that we only output a transcript if $X\times Y$ is $\rho$-structured and $\Dmin(\Y)\le n^3$ at the end, which by \autoref{lem:general} guarantees that the transcript's rectangle intersects the slice $G^{-1}(z)$ and thus $\t\in\supp(\t')$.

\subsection{Efficiency: Number of queries}

We show that our randomized decision tree makes $O(|\Pi|/\log n)$ queries with high probability. If we insist on a decision tree that \emph{always} makes this many queries (to match the statement of \autoref{thm:goal}), we may terminate the execution early (with output $\bot$) whenever we exceed the threshold. This would incur only a small additional loss in the closeness of transcript distributions.

\begin{lemma}
The simulation makes $O(|\Pi|/\log n)$ queries with probability $\ge 1-\min(2^{-|\Pi|},1/n^{\Omega(1)})$.
\end{lemma}
\begin{proof}
During the simulation, we view the quantity $\Dmin(\X_{\free\rho})\geq 0$ as a nonnegative potential function. Consider a single iteration where lines 11, 14, 16 modify the sets $X$ and $\free\rho$.
\begin{itemize}[itemsep=2mm]
\item In line~11, we shrink $X=X^0\cup X^1$ down to $X^{\b}$ where $\Pr[\b=b]=|X^b|/|X|$. Hence the increase in the potential function is $\gamma_{\b}\coloneqq \log(|X|/|X^{\b}|)$.
\item In line 14 (after $X\leftarrow X^b$), we shrink $X=\bigcup_i X^i$ down to $X^{\i}$ where $\Pr[\i=i]=|X^i|/|X|$. Moreover, in line 16, $|\!\free\rho|$ decreases by the number of bits we query. \autoref{lem:density-restoring} says that the potential changes by $\delta_{\i} -\Omega(\log n)\cdot\textbf{\#}(\text{queries in this iteration})$ where $\delta_{\i}\coloneqq\log(|X|/|\cup_{j\geq \i}X^j|)$.
\end{itemize}
We will see later that for any iteration, $\Ex[\gamma_{\b}],\Ex[\delta_{\i}]\le O(1)$.

For $j=1,\ldots,|\Pi|$, letting $\bgamma_j,\bdelta_j$ be the random variables $\gamma_{\b},\delta_{\i}$ respectively in the $j$-th iteration (and letting $\bgamma_j=\bdelta_j=0$ for outcomes in which Alice does not communicate in the $j$-th iteration), the potential function at the end of the simulation is $\sum_j(\bgamma_j+\bdelta_j)-\Omega(\log n)\cdot\textbf{\#}(\text{queries in total})\ge 0$ and hence \[\textstyle\Ex\bigl[\textbf{\#}(\text{queries in total})\bigr]~\le~O(1/\log n)\cdot\sum_j\bigl(\Ex[\bgamma_j]+\Ex[\bdelta_j]\bigr)~\le~O(|\Pi|/\log n).\] By Markov's inequality, this already suffices to show that with probability $\ge 0.9$ (say), the simulation uses $O(|\Pi|/\log n)$ queries. To get a better concentration bound, we would like for the $\bgamma_j,\bdelta_j$ variables (over all $j$) to be mutually independent, which they unfortunately generally are not. However, there is a trick to overcome this: we will define mutually independent random variables $\c_j,\d_j$ (for all $j$) and couple them with the $\bgamma_j,\bdelta_j$ variables in such a way that each $\bgamma_j\le\c_j$ and $\bdelta_j\le\d_j$ with probability $1$, and show that $\sum_j(\c_j+\d_j)$ is bounded with very high probability, which implies the same for $\sum_j(\bgamma_j+\bdelta_j)$. For each $j$, do the following.
\begin{itemize}
\item Sample a uniform real $\p_j\in[0,1)$ and define $\c_j\coloneqq\log(1/\p_j)+\log(1/(1-\p_j))$ and let $\bgamma_j\coloneqq\gamma_{\b}$ where $\b=0$ if $\p_j\in[0,|X^0|/|X|)$ and $\b=1$ if $\p_j\in[|X^0|/|X|,1)$ (where $X,X^0,X^1$ are the sets that arise in the first half of the $j$-th iteration, conditioned on the outcomes of previous iterations). Note that $\bgamma_j$ is correctly distributed, and that $\bgamma_j\le\c_j$ with probability $1$ (specifically, if $\b=0$ then $\bgamma_j=\log(|X|/|X^0|)\le\log(1/\p_j)\le\c_j$ and if $\b=1$ then $\bgamma_j=\log(|X|/|X^1|)\le\log(1/(1-\p_j))\le\c_j$). Also note that, as claimed earlier, $\Ex[\bgamma_j]\le\Ex[\c_j]=\int_0^1\bigl(\log(1/p)+\log(1/(1-p))\bigr)\,\mathrm{d}p=2/\ln 2\le O(1)$. For future use, note that $\Ex\bigl[2^{\c_j/2}\bigr]=\int_0^1(p(1-p))^{-1/2}\,\mathrm{d}p=\pi\le O(1)$.\vspace{8pt}
\item Sample a uniform real $\q_j\in[0,1)$ and define $\d_j\coloneqq\log(1/(1-\q_j))$ and let $\bdelta_j\coloneqq\delta_{\i}$ where $\i$ is such that $\q_j$ falls in the $\i$-th interval, assuming we have partitioned $[0,1)$ into half-open intervals with lengths $|X^i|/|X|$ in the natural left-to-right order (where $X,X^1,X^2,\ldots$ are the sets that arise in the second half of the $j$-th iteration, conditioned on the outcomes of the first half and previous iterations). Note that $\bdelta_j$ is correctly distributed, and that $\bdelta_j\le\d_j$ with probability $1$ (specifically, if $\i=i$ then $\bdelta_j=\log(|X|/|\cup_{j\ge i}X^j|)\le\log(1/(1-\q_j))=\d_j$). Also note that, as claimed earlier, $\Ex[\bdelta_j]\le\Ex[\d_j]\le\Ex[\c_j]\le O(1)$. For future use, note that $\Ex\bigl[2^{\d_j/2}\bigr]\le\Ex\bigl[2^{\c_j/2}\bigr]\le O(1)$.
\end{itemize}
Now for some sufficiently large constants $C,C'$ we have
{\allowdisplaybreaks
\begin{align*}
\Pr\bigl[\textbf{\#}(\text{queries in total})>C'\cdot|\Pi|/\log n\bigr]~&\textstyle\le~\Pr\bigl[\sum_j(\bgamma_j+\bdelta_j)>C\cdot|\Pi|\bigr]\\
&\textstyle\le~\Pr\bigl[\sum_j(\c_j+\d_j)>C\cdot|\Pi|\bigr]\\
&=~\Pr\Bigl[2^{\sum_j(\c_j+\d_j)/2}>2^{C\cdot|\Pi|/2}\Bigr]\\
&\le~\Ex\bigl[2^{\sum_j(\c_j+\d_j)/2}\bigr]/2^{C\cdot|\Pi|/2}\\
&\textstyle=~\Bigl(\prod_j\Ex\bigl[2^{\c_j/2}\bigr]\cdot\Ex\bigl[2^{\d_j/2}\bigr]\Bigr)/2^{C\cdot|\Pi|/2}\\
&\le~\bigl(O(1)/2^{C/2}\bigr)^{|\Pi|}\\
&\le~2^{-|\Pi|}.
\end{align*}}
If $|\Pi|\le o(\log n)$ then a similar calculation shows that $\Pr\bigl[\textbf{\#}(\text{queries in total})\ge 1\bigr]\le 1/n^{\Omega(1)}$.
\end{proof}

\section{Uniform Marginals Lemma} \label{sec:pseudorandomness}

\unifmarg*

We prove a slightly stronger statement formulated in \autoref{lem:pointwise-unif} below. For terminology, we say a distribution $\calD_1$ is \emph{$\varepsilon$-pointwise-close} to a distribution $\calD_2$ if for every outcome, the probability under $\calD_1$ is within a factor $1\pm\varepsilon$ of the probability under $\calD_2$. As a minor technicality (for the purpose of deriving \autoref{lem:general} from \autoref{lem:pointwise-unif}), we say that a random variable $\x\in[m]^J$ is \emph{$\delta$-essentially-dense} if for every nonempty $I\subseteq J$, $\Hmin(\x_I) \geq \delta\cdot |I|\log m-1$ (the difference from \autoref{def:dense} is the ``$-1$''); we also define \emph{$\rho$-essentially-structured} in the same way as $\rho$-structured but requiring $\X_{\free\rho}$ to be only $0.9$-essentially-dense instead of $0.9$-dense. The following strengthens a lemma from \cite{goos17query}, which implied that $G(\X,\Y)$ has full support over the set of all $z$ consistent with $\rho$.

\begin{lemma}[Pointwise uniformity] \label{lem:pointwise-unif}
Suppose $X\times Y$ is $\rho$-essentially-structured and $\Dmin(\Y)\leq n^3+1$. Then $G(\X,\Y)$ is $1/n^3$-pointwise-close to the uniform distribution over the set of all $z$ consistent with $\rho$.
\end{lemma}

\begin{proof}[Proof of \autoref{lem:general}]
Let $(\x,\y)$ be uniformly distributed over $G^{-1}(z) \cap X \times Y$. We show that $\x$ is $1/n^2$-close to $\X$; a completely analogous argument works to show that $\y$ is $1/n^2$-close to $\Y$. Let $E \subseteq X$ be any test event. Replacing $E$ by $X\smallsetminus E$ if necessary, we may assume $|E| \geq |X|/2$. Since $X\times Y$ is $\rho$-structured, $E\times Y$ is $\rho$-essentially-structured. Hence we can apply \autoref{lem:pointwise-unif} in both the rectangles $E\times Y$ and $X \times Y$:
\begin{align*}
\Pr[\x\in E]~&=~\frac{|G^{-1}(z) \cap E \times Y|}{|G^{-1}(z) \cap X\times Y|}~=~\frac{(1\pm 1/n^3)\cdot 2^{-\left|\free\rho\right|}\cdot|E \times Y|}{(1\pm 1/n^3)\cdot 2^{-\left|\free\rho\right|}\cdot|X\times Y|}\\
&=~(1\pm 3/n^3)\cdot|E|/|X|~=~|E|/|X|\pm 1/n^2.\qedhere
\end{align*}
\end{proof}

\subsection{Overview for \autoref{lem:pointwise-unif}}

A version of \autoref{lem:pointwise-unif} (for the inner-product gadget) was proved in \cite[\S2.2]{goos16rectangles} under the assumption that $\X$ and $\Y$ had low deficiencies: $\Dmin(\X_I),\Dmin(\Y_I)\leq O(|I|\log n)$ for free blocks~$I$. The key difference is that we only assume $\Dmin(\Y_I)\leq n^3+1$. We still follow the general plan from~\cite{goos16rectangles} but with a new step that allows us to reduce the deficiency of $\Y$.

\paragraph{Fourier perspective.}
The idea in \cite{goos16rectangles} to prove that $\z\coloneqq G(\X,\Y)$ is pointwise-close to uniform is to study $\z$ in the Fourier domain, and show that $\z$'s Fourier coefficients (corresponding to free blocks) decay exponentially fast. That is, for every nonempty $I\subseteq\free\rho$ we want to show that the bias of $\oplus(\z_I)$ (parity of the output bits $\z_I$) is exponentially small in $|I|$. Tools tailor-made for this situation exist: various ``$\Xor$ lemmas'' are known to hold for communication complexity (e.g.,~\cite{shaltiel03towards}) that apply as long as $\X_I$ and $\Y_I$ have low deficiencies. All this is recalled in \autoref{sec:fourier}. This suggests that all that remains is to reduce our case of high deficiency (of $\Y_I$) to the case of low deficiency.

\paragraph{Reducing deficiency via buckets.}
For the moment assume $I=[n]$ for simplicity of discussion. Our idea for reducing the deficiency of $\Y_I=\Y$ is as follows. We partition each $m$-bit string in $\Y\in(\{0,1\}^m)^n$ into $m^{1/2}$ many \emph{buckets} each of length $m^{1/2}$. We argue that $\Y$ can be expressed as a mixture of distributions $\y$, where $\y$ has few of its buckets fixed in each string $\y_i$, and for any way of choosing an unfixed bucket for each $\y_i$, the marginal distribution of $\y$ on the union $T$ of these buckets has deficiency as low as $\Dmin(\y_T)\le 1$. Correspondingly, we argue that $\X$ may be expressed as a mixture of distributions $\x$ that have a nice form:
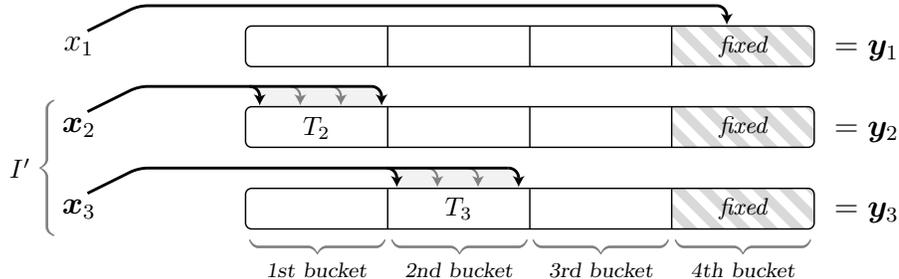
\begin{center}
\begin{lpic}[b(3mm),t(1mm),r(7mm),l(-7mm)]{bucket(.27)}
\large
\lbl[c]{38,110;$x_1$}
\lbl[c]{38,70;$\x_2$}
\lbl[c]{38,30;$\x_3$}
\lbl[l]{410,108;$=\y_1$}
\lbl[l]{410,68;$=\y_2$}
\lbl[l]{410,28;$=\y_3$}

\small
\lbl[c]{9,50; $I'$}

\lbl[c]{155,70;$T_2$}
\lbl[c]{225,30;$T_3$}

\footnotesize\itshape
\lbl[c]{365,110;fixed}
\lbl[c]{365,70;fixed}
\lbl[c]{365,30;fixed}

\scriptsize\slshape
\lbl[c]{155,0;1st bucket}
\lbl[c]{225,0;2nd bucket}
\lbl[c]{295,0;3rd bucket}
\lbl[c]{365,0;4th bucket}
\end{lpic}
\end{center}
Here each pointer $\x_i$ ranges over a single bucket $T_i$. Moreover, for a large subset $I'\subseteq[n]$ of coordinates, $T_i$ is unfixed in $\y_i$ for $i\in I'$, and hence $\y$ has deficiency $\le 1$ on the union of these unfixed buckets. The remaining few $i\in[n]\smallsetminus I'$ are associated with fixed pointers $\x_i=x_i$ pointing into fixed buckets in $\y$. Consequently, we may interpret $(\x,\y)$ as a random input to $\Ind_{m^{1/2}}^n$ by identifying each bucket $T_i$ with $[m^{1/2}]$. In this restricted domain, we can show that $(\oplus\circ g^n)(\x,\y)$ is indeed very unbiased: the fixed coordinates do not contribute to the bias of the parity, and $(\x_{I'},\y_{I'})$ is a pair of low-deficiency variables for which an $\Xor$ lemma type calculation applies. The heart of the proof will be to find a decomposition of $\X\times\Y$ into such distributions $\x\times\y$.

In the remaining subsections, we carry out the formal proof of \autoref{lem:pointwise-unif}.

\subsection{Fourier perspective} \label{sec:fourier}

Henceforth we abbreviate $J\coloneqq\free\rho$. We employ the following calculation from \cite{goos16rectangles}, whose proof is reproduced in \autoref{sec:fourier-proof} for completeness. Here $\chi(z)\coloneqq(-1)^{\oplus(z)}$.

\begin{restatable}[Pointwise uniformity from parities]{lemma}{fourier} \label{lem:fourier}
If a random variable $\z_J$ over $\{0,1\}^J$ satisfies $\bigl|\Ex\bigl[\chi(\z_I)\bigr]\bigr|\le 2^{-5|I|\log n}$ for every nonempty $I\subseteq J$, then $\z_J$ is $1/n^3$-pointwise-close to uniform.
\end{restatable}

To prove \autoref{lem:pointwise-unif}, it suffices to take $\z_J=g^J(\X_J,\Y_J)$ above and show for every $\emptyset\neq I\subseteq J$,
\begin{equation} \label{eq:bias}
\bigl|\Ex\bigl[\chi(g^I(\X_I,\Y_I))\bigr]\bigr|~\le~2^{-5|I|\log n}.
\end{equation}
In our high-deficiency case, we have
\begin{enumerate}
\item[(i)] $\Dmin(\X_I)\le 0.1|I|\log m+1$,
\item[(ii)] $\Dmin(\Y_I)\le n^3+1$.
\end{enumerate}

\paragraph{Low-deficiency case.}
As a warm-up, let us see how to obtain \eqref{eq:bias} by imagining that we are in the low-deficiency case, i.e., replacing assumption (ii) by
\begin{enumerate}
\item[(ii$'$)] $\Dmin(\Y_I)\le 1$.
\end{enumerate}
We present a calculation that is a very simple special case of, e.g., Shaltiel's~\cite{shaltiel03towards} $\Xor$ lemma for discrepancy (relative to uniform distribution).

Let $M$ be the communication matrix of $g\coloneqq\Ind_m$ but with $\{+1,-1\}$ instead of $\{0,1\}$ entries. The operator $2$-norm of $M$ is $\|M\|=2^{m/2}$ since the rows are orthogonal and each has $2$-norm $2^{m/2}$. The $|I|$-fold tensor product of $M$ then satisfies $\bigl\|M^{\otimes|I|}\bigr\|=2^{|I|m/2}$ by the standard fact that the $2$-norm behaves multiplicatively under tensor product. Here $M^{\otimes|I|}$ is the communication matrix of the 2-party function $\chi\circ g^I$. We think of the distribution of $\X_I$ as an $m^{|I|}$-dimensional vector $\calD_{\X_I}$, and of the distribution of $\Y_I$ as a $(2^m)^{|I|}$-dimensional vector $\calD_{\Y_I}$. Letting $\H_2$ ($\ge\Hmin$) denote R{\'e}nyi $2$-entropy, by (i) we have
\[
\bigl\|\calD_{\X_I}\bigr\|~=~2^{-\H_2(\X_I)/2}~\le~2^{-\Hmin(\X_I)/2}~\le~2^{-(|I|\log m-0.1|I|\log m-1)/2}~=~2^{-0.45|I|\log m+1/2}.
\]
Similarly, by (ii$'$) we would have
\[
\bigl\|\calD_{\Y_I}\bigr\|~\le~2^{-(|I|m-1)/2}~=~2^{-|I|m/2+1/2}.
\]
The left side of \eqref{eq:bias} is now
\begin{align}
\Bigl|\calD_{\X_I}^\top\,M^{\otimes|I|}\,\calD_{\Y_I}\Bigr|~\le~\bigl\|\calD_{\X_I}\bigr\|\cdot\bigl\|M^{\otimes|I|}\bigr\|\cdot\bigl\|\calD_{\Y_I}\bigr\|~&\le~2^{-0.45|I|\log m+1/2}\cdot 2^{|I|m/2}\cdot 2^{-|I|m/2+1/2}\notag\\
&=~2^{-0.45|I|\log m+1}~\le~2^{-5|I|\log n}.\label{eq:norm}
\end{align}
Therefore our goal becomes to reduce (via buckets) from case (ii) to case (ii$'$).

\subsection{Buckets}

We introduce some bucket terminology for random $(\x,\y)\in[m]^I\times(\{0,1\}^m)^I$.
\begin{itemize}[itemsep=2mm]
\item Each string $\y_i$ is partitioned into $m^{1/2}$ buckets each of length $m^{1/2}$.
\item We think of $\x_i$ as a pair $\bell_i\r_i$ where $\bell_i$ specifies which bucket and $\r_i$ specifies which element of the bucket. (Or, viewing $\x_i\in\{0,1\}^{\log m}$, $\bell_i\in\{0,1\}^{(\log m)/2}$ would be the left half and $\r_i\in\{0,1\}^{(\log m)/2}$ would be the right half.) Thus $\x=\bell\r$ where the random variable $\bell\in[m^{1/2}]^I$ picks a bucket for each coordinate, and the random variable $\r\in[m^{1/2}]^I$ picks an element from each of the buckets specified by $\bell$. Every outcome $\ell$ of $\bell$ has an associated \emph{bucket union} (one bucket for each string) given by $T_\ell\coloneqq\bigcup_{i\in I}(\{i\}\times T_{\ell_i})$ where $T_{\ell_i}\subseteq[m]$ is the bucket specified by $\ell_i$. Here a bit index $(i,j)\in I\times[m]$ refers to the $j$-th bit of the string $\y_i$.
\end{itemize}

\subsection{Focused decompositions}
Our goal is to express the product distribution $\X_I\times\Y_I$ as a convex combination of product distributions $\x\times\y$ that are \emph{focused}, which informally means that many pointers in $\x$ point into buckets that collectively have low deficiency in $\y$, and the remaining pointers produce constant gadget outputs. A formal definition follows.
\begin{definition} \label{def:focused}
A product distribution $\x\times\y$ over $[m]^I\times(\{0,1\}^m)^I$ is called \emph{focused} if there is a partial assignment $\sigma\in\{0,1,*\}^I$ such that, letting $I'\coloneqq\free\sigma$, we have: $|I'|\ge|I|/2$, and $g^I(\x,\y)$ is always consistent with $\sigma$, and for each $i\in I'$, $\x_i=\ell_i\r_i$ is always in a specific bucket $T_{\ell_i}\subseteq[m]$, and
\begin{itemize}
\item[(i$^*$)] $\Dmin(\x_{I'})\le 0.6|I'|\log m^{1/2}$~~with respect to $\bigtimes_{i\in I'}T_{\ell_i}$,
\item[(ii$^*$)] $\Dmin(\y_T)\le 1$~~where $T\coloneqq\bigcup_{i\in I'}(\{i\}\times T_{\ell_i})$.
\end{itemize}
\end{definition}

We elaborate on this definition. Since $g^I(\x,\y)$ is always consistent with $\sigma$, the coordinates $\fix\sigma=I\smallsetminus I'$ are irrelevant to the bias of the parity of $g^I(\x,\y)$. For each $i\in I'$, we might as well think of the domain of $\x_i$ as $T_{\ell_i}$ instead of $[m]$, and of the domain of $\y_i$ as $\{0,1\}^{T_{\ell_i}}$ instead of $\{0,1\}^m$. Hence, out of the $|I'|m$ bits of $\y_{I'}$, the only relevant ones are the $|I'|m^{1/2}$ bits indexed by $T$. We may thus interpret $(\x_{I'},\y_T)$ as a random input to $\Ind_{m^{1/2}}^{I'}$. In summary,
\begin{equation} \label{eq:interpret}
\bigl|\Ex\bigl[\chi(g^I(\x,\y))\bigr]\bigr|~=~\bigl|\Ex\bigl[\chi(g^{I'}(\x_{I'},\y_{I'}))\bigr]\bigr|~=~\bigl|\Ex\bigl[\chi(\Ind_{m^{1/2}}^{I'}(\x_{I'},\y_T))\bigr]\bigr|.
\end{equation}

If $\x\times\y$ is focused, then the calculation leading to \eqref{eq:norm} can be applied to $\x_{I'}\times\y_T$ with $m$ replaced by $m^{1/2}$, $|I|$ replaced by $|I'|\ge|I|/2$, and min-entropy rate $0.9$ replaced by $0.4$, to show that
\begin{equation*}
\text{value of \eqref{eq:interpret}}~\le~2^{-0.2|I'|\log m^{1/2}+1}~\le~2^{-(0.2/4)|I|\log m+1}~\le~2^{-5|I|\log n-1}. \tag{using $m=n^{256}$}
\end{equation*}

\begin{lemma} \label{lem:reduce}
The product distribution $\X_I\times\Y_I$ can be decomposed into a mixture of product distributions $\Ex_{d\sim\d}[\x^d\times\y^d]$ over $[m]^I\times(\{0,1\}^m)^I$ ($d$ stands for ``data'') such that $\x^d\times\y^d$ is focused with probability at least $1-2^{-5|I|\log n-1}$ over $d\sim\d$.
\end{lemma}

Using \autoref{lem:reduce}, which we prove in the following subsection, we can derive \eqref{eq:bias}:
\begin{align*}
\bigl|\Ex\bigl[\chi(g^I(\X_I,\Y_I))\bigr]\bigr|~&\le~\Ex_{d\sim\d}\bigl|\Ex\bigl[\chi(g^I(\x^d,\y^d))\bigr]\bigr|\\
&\le~\Pr[\d\text{ is not focused}]+\max_{\text{focused }d}\bigl|\Ex\bigl[\chi(g^I(\x^d,\y^d))\bigr]\bigr|\\
&\le~2^{-5|I|\log n-1}+2^{-5|I|\log n-1}~=~2^{-5|I|\log n}.
\end{align*}

\subsection{Finding a focused decomposition}

We now prove \autoref{lem:reduce}. By assumption, $\X_I=\bell\r$ is $0.9$-essentially-dense (since $\X_J$ is) and $\Dmin(\Y_I)\le\Dmin(\Y)\le n^3+1$. We carry out the decomposition in the following three steps. Define $\varepsilon\coloneqq 2^{-5|I|\log n-1}$.

\begin{claim} \label{clm:step1}
$\Y_I$ can be decomposed into a mixture of distributions $\Ex_{c\sim\c}[\y^c]$ over $(\{0,1\}^m)^I$ such that with probability at least $1-\varepsilon/3$ over $c\sim\c$,
\begin{enumerate}[label=(P\arabic*),leftmargin=13mm]
\item \label{p1} each string in $\y^c$ has at most $2n^3$ fixed buckets,
\item \label{p2} each bucket union $T_\ell$ not containing fixed buckets has $\Dmin(\y^c_{T_\ell})\le 1$.
\end{enumerate}
\end{claim}

\begin{claim} \label{clm:step2}
For any $c$ satisfying \ref{p1}, with probability at least $1-\varepsilon/3$ over $\ell\sim\bell$,
\begin{enumerate}[label=(Q\arabic*),leftmargin=13mm]
\item \label{q1} the bucket union $T_\ell$ contains at most $|I|/2$ fixed buckets of $\y^c$,
\item \label{q2} $\Dmin(\r\mid\bell=\ell)\le 0.25|I|\log m^{1/2}$.
\end{enumerate}
\end{claim}

\begin{claim} \label{clm:step3}
For any $c$ and $\ell$ satisfying \ref{q1}, \ref{q2}, letting \[I^*~\coloneqq~\bigl\{i\in I\,:\,\text{the $\ell_i$ bucket of $\y^c_i$ is fixed}\bigr\}\qquad\text{and}\qquad I'~\coloneqq~I\smallsetminus I^*,\] with probability at least $1-\varepsilon/3$ over $r_{I^*}\sim(\r_{I^*}\mid\bell=\ell)$, we have $\Dmin(\r_{I'}\mid\bell=\ell,\,\r_{I^*}=r_{I^*})\le 0.6|I'|\log m^{1/2}$.
\end{claim}

We now finish the proof of \autoref{lem:reduce} assuming these three claims. Take $\d\coloneqq\bigl(\c,\bell,\r_{\I^*})$; that is, the data $d\sim\d$ is sampled by first sampling $c\sim\c$, then $\ell\sim\bell$, then $r_{I^*}\sim(\r_{I^*}\mid\bell=\ell)$, where $I^*$ implicitly depends on $c$ and $\ell$. Take $\y^d\coloneqq\y^c$ and $\x^d\coloneqq(\X_I\mid\bell=\ell,\,\r_{I^*}=r_{I^*})$, and note that $\Ex_{d\sim\d}[\x^d\times\y^d]$ indeed forms a decomposition of $\X_I\times\Y_I$. By a union bound, with probability at least $1-\varepsilon$ over $d\sim\d$, the properties of all three claims hold, in which case we just need to check that $\x^d\times\y^d$ is focused.

Since for each $i\in I^*$, $\x^d_i\in T_{\ell_i}$ and $\y^d_{i,T_{\ell_i}}$ are both fixed, we have that $g^{I^*}(\x^d_{I^*},\y^d_{I^*})$ is fixed and hence $g^I(\x^d,\y^d)$ is always consistent with some partial assignment $\sigma$ with $\fix\sigma=I^*$ and $\free\sigma=I'$. We have $|I'|\ge|I|/2$ by \emph{\ref{q1}}. For each $i\in I'$, note that $\x^d_i$ is always in $T_{\ell_i}$ since we conditioned on $\bell=\ell$. Note that (i$^*$) for $\x^d$ holds by \autoref{clm:step3}. To see that (ii$^*$) for $\y^d$ holds, pick any $\ell'$ that agrees with $\ell$ on $I'$ and such that for every $i\in I^*$, the $\ell'_i$ bucket of $\y^d_i$ is not fixed---thus, the bucket union $T_{\ell'}$ contains no fixed buckets of $\y^d$---and note that $\Dmin(\y^d_T)\le\Dmin(\y^d_{T_{\ell'}})\le 1$ by \emph{\ref{p2}}.

\begin{proof}[Proof of \autoref{clm:step1}]
We use a process highly reminiscent of the ``density-restoring partition'' process described in \autoref{sec:density-restoring}. We maintain an event $E$ which is initially all of $(\{0,1\}^m)^I$.\bigskip

\noindent
While $\Pr[\Y_I\in E]>\varepsilon/3$:
\begin{enumerate}[topsep=1mm]
\item Choose a maximal set of pairwise disjoint bucket unions $\calT = \{T_{\ell^1},\ldots,T_{\ell^k}\}$ with the property that $\Dmin(\Y_{\cup\calT}\mid E)>k$ (possibly $\calT=\emptyset$) and let $\beta\in\{0,1\}^{\cup\calT}$ be an outcome witnessing this: $\Pr[\Y_{\cup\calT}=\beta\mid E]>2^{-(k|I|m^{1/2}-k)}$.
\item Output the distribution $(\Y_I\mid\Y_{\cup\calT}=\beta,\,E)$ with associated probability $\Pr[\Y_{\cup\calT}=\beta,\,E]>0$.
\item Update $E\leftarrow\bigl\{y_I\in E\,:\,y_{\cup\calT}\ne\beta\bigr\}$.
\end{enumerate}
Output the distribution $(\Y_I\mid E)$ with associated probability $\Pr[\Y_I\in E]$ if the latter is nonzero.\bigskip

\noindent
The distributions output throughout the process are the $\y^c$'s; note that with the associated probabilities, they indeed form a decomposition of $\Y_I$. Each time \emph{(1)} is executed, we have \[k~<~\Dmin(\Y_{\cup\calT}\mid E)~\le~\Dmin(\Y_I)+\log(1/\Pr[\Y_I\in E])~\le~n^3+1+\log(3/\varepsilon)~\le~2n^3.\] Also, any $\y^c=(\Y_I\mid\Y_{\cup\calT}=\beta,\,E)$ output in \emph{(2)} has the property that for any bucket union $T_\ell$ not containing fixed buckets, $\Dmin(\y^c_{T_\ell})\le 1$. To see this, first note that $T_\ell$ is disjoint from $\cup\calT$ since the latter buckets are fixed to $\beta$. If $\Dmin(\y^c_{T_\ell})>1$ were witnessed by some $\gamma\in\{0,1\}^{T_\ell}$, then
\begin{align*}
\Pr[\Y_{(\cup\calT)\cup T_\ell}=\beta\gamma\mid E]~&=~\Pr[\Y_{\cup\calT}=\beta\mid E]\cdot\Pr[\Y_{T_\ell}=\gamma\mid\Y_{\cup\calT}=\beta,\,E]\\
&>~2^{-(k|I|m^{1/2}-k)}\cdot 2^{-(|I|m^{1/2}-1)}~=~2^{-((k+1)|I|m^{1/2}-(k+1))}
\end{align*}
and so $\Dmin(\Y_{(\cup\calT)\cup T_\ell}\mid E)>k+1$, which would contradict the maximality of $k$ since $\{T_{\ell^1},\ldots,T_{\ell^k},T_\ell\}$ is a set of pairwise disjoint bucket unions.
\end{proof}

\begin{proof}[Proof of \autoref{clm:step2}]
Assume that for each coordinate $i\in I$, $\y^c_i$ has at most $2n^3$ fixed buckets. Since $\X_I$ is $0.9$-essentially-dense, $\bell$ is $0.8$-essentially-dense (for each nonempty $H\subseteq I$, we have $\Dmin(\bell_H)\le\Dmin(\X_H)\le 0.1|H|\log m+1=0.2|H|\log m^{1/2}+1$). Thus, the probability that $T_{\bell}$ hits fixed buckets in all coordinates in some set $H\subseteq I$ is at most the number of ways of choosing a fixed bucket from each of those coordinates ($\le (2n^3)^{|H|}$) times the maximum probability that $T_{\bell}$ hits all the chosen buckets ($\le 2^{-(0.8|H|\log m^{1/2}-1)}$ since $\bell$ is $0.8$-essentially-dense). We can now calculate
\begin{align*}
\Pr[T_{\bell}\text{ hits $\ge|I|/2$ fixed buckets}]~
&\le~\textstyle\sum_{H\subseteq I,|H|=|I|/2}\Pr[T_{\bell}\text{ hits fixed buckets in coordinates $H$}]\\
&\le~\textstyle\binom{|I|}{|I|/2}\cdot (2n^3)^{|I|/2}\cdot 2^{-(0.8(|I|/2)\log m^{1/2}-1)}\\
&\le~2^{|I|}\cdot 2^{1.5|I|\log n+1} \cdot 2^{-(51.2|I|\log n-1)}
\tag{using $m= n^{256}$}\\
&\le~2^{|I|-49.7|I|\log n+2}\\
&\le~\varepsilon/6
\end{align*}
For convenience, we assumed above that $|I|$ is even; if $|I|$ is odd (including the case $|I|=1$), the same calculation works with $\lceil|I|/2\rceil$ instead of $|I|/2$.

\emph{\ref{q2}} follows by a direct application of the chain rule for min-entropy \cite[Lemma 6.30]{vadhan12pseudorandomness}: with probability at least $1-\varepsilon/6$ over $\ell\sim\bell$, we have \[\Dmin(\r\mid\bell=\ell)~\le~\Dmin(\X_I)+\log(6/\varepsilon)~\le~\bigl(0.1|I|\log m+1\bigr)+\bigl(5|I|\log n+4\bigr)~\le~0.25|I|\log m^{1/2}.\] By a union bound, with probability at least $1-\varepsilon/3$ over $\bell$, \emph{\ref{q1}} and \emph{\ref{q2}} hold simultaneously.
\end{proof}

\begin{proof}[Proof of \autoref{clm:step3}]
This is again a direct application of the chain rule for min-entropy: with probability at least $1-\varepsilon/3$ over $r_{I^*}\sim(\r_{I^*}\mid\bell=\ell)$, we have
\begin{align*}
\Dmin(\r_{I'}\mid\bell=\ell,\,\r_{I^*}=r_{I^*})~&\le~\Dmin(\r\mid\bell=\ell)+\log(3/\varepsilon)\\
&\le~\bigl(0.25|I|\log m^{1/2}\bigr)+\bigl(5|I|\log n+3\bigr)~\le~0.6|I'|\log m^{1/2}
\end{align*}
where the middle inequality uses \emph{\ref{q2}}, and the last inequality uses \emph{\ref{q1}} ($|I'|\ge|I|/2$) and $m=n^{256}$.
\end{proof}

\subsection{Pointwise uniformity from parities} \label{sec:fourier-proof}

\fourier*

\begin{proof}[Proof (from {\upshape\cite[\S2.2]{goos16rectangles}})]
We let $\varepsilon\coloneqq 1/n^3$ and write $\z_J$ as $\z$ throughout the proof. We think of the distribution of $\z$ as a function $\calD\colon\{0,1\}^J\to[0,1]$ and write it in the Fourier basis as \[\textstyle\calD(z)~=~\sum_{I\subseteq J}\widehat{\calD}(I)\chi_I(z)\] where $\chi_I(z)\coloneqq(-1)^{\oplus(z_I)}$ and $\widehat{\calD}(I)\coloneqq 2^{-|J|}\sum_{z}\calD(z)\chi_I(z)=2^{-|J|}\cdot\Ex[\chi_I(\z)]$. Note that $\widehat{\calD}(\emptyset)=2^{-|J|}$ because $\calD$ is a distribution. Our assumption says that for all nonempty $I\subseteq J$, $2^{|J|}\cdot|\widehat{\calD}(I)|\le 2^{-5|I|\log n}$, which is at most $\varepsilon 2^{-2|I|\log|J|}$. Hence, \[\textstyle 2^{|J|}\sum_{I\neq\emptyset}|\widehat{\calD}(I)|~\le~\varepsilon\sum_{I\neq\emptyset}2^{-2|I|\log|J|}~=~\varepsilon\sum_{k=1}^{|J|}\binom{|J|}{k}2^{-2k\log|J|}~\le~\varepsilon\sum_{k=1}^{|J|}2^{-k\log|J|}~\le~\varepsilon.\] We use this to show that $\bigl|\calD(z)-2^{-|J|}\bigr|\le\varepsilon 2^{-|J|}$ for all $z\in\{0,1\}^J$, which proves the lemma. To this end, let $\calU$ denote the uniform distribution (note that $\widehat{\calU}(I)=0$ for all nonempty $I\subseteq J$) and let $\mathds{1}_z$ denote the indicator for $z$ defined by $\mathds{1}_z(z)=1$ and $\mathds{1}_z(z')=0$ for $z'\ne z$ (note that $|\widehat{\mathds{1}}_z(I)|=2^{-|J|}$ for all $I$). We can now calculate
\begin{align*}
\bigl|\calD(z)-2^{-|J|}\bigr|~&=~\bigl|\langle\mathds{1}_z,\calD\rangle-\langle\mathds{1}_z,\calU\rangle\bigr|~=~|\langle\mathds{1}_z,\calD-\calU\rangle|~=~2^{|J|}\cdot|\langle\widehat{\mathds{1}}_z,\widehat{\calD}-\widehat{\calU}\rangle|\\
&\le~2^{|J|}\cdot{\textstyle\sum_{I\neq\emptyset}}|\widehat{\mathds{1}}_z(I)|\cdot|\widehat{\calD}(I)|~=~{\textstyle\sum_{I\neq\emptyset}}|\widehat{\calD}(I)|~\le~\varepsilon 2^{-|J|}.\qedhere
\end{align*}
\end{proof}

\section{Applications} \label{sec:applications}

In this section, we collect some recent results in communication complexity, which we can derive (often with simplifications) from our lifting theorem.

\paragraph{Classical vs.\ quantum.}
Anshu et al.~\cite{anshu16separations} gave a nearly 2.5-th power total function separation between quantum and classical randomized protocols. Our lifting theorem can reproduce this separation by lifting an analogous separation in query complexity due to Aaronson, Ben-David, and Kothari~\cite{aaronson16separations}. Let us also mention that Aaronson and Ambainis~\cite{aaronson15forrelation} have conjectured that a slight generalization of $\Forr$ witnesses an $O(\log n)$-vs-$\tOmega(n)$ quantum/classical query separation. If true, our lifting theorem implies that~``2.5'' can be improved to ``$3$'' above; see~\cite{aaronson16separations} for a discussion. (Such an improvement is not black-box implied by the techniques of Anshu et al.~\cite{anshu16separations}.)

Raz~\cite{raz99exponential} gave an exponential partial function separation between quantum and classical randomized protocols. Our lifting theorem can reproduce this separation by lifting, say, the $\Forr$ partial function~\cite{aaronson15forrelation}, which witnesses a $1$-vs-$\tOmega(\sqrt{n})$ separation for quantum/classical query complexity. However, qualitatively stronger separations are known \cite{klartag11quantum,gavinsky16entangled} where the quantum protocol can be taken to be \emph{one-way} or even \emph{simultaneous}.

\paragraph{Partition numbers.}
Anshu et al.~\cite{anshu16separations} gave a nearly quadratic separation between (the log of) the \emph{two-sided partition number} (number of monochromatic rectangles needed to partition the domain of $F$) and randomized communication complexity. This result now follows by lifting an analogous separation in query complexity due to Ambainis, Kokainis, and Kothari~\cite{ambainis16nearly}.

In \cite{goos15randomized}, a nearly quadratic separation was shown between (the log of) the \emph{one-sided partition number} (number of rectangles needed to partition $F^{-1}(1)$) and randomized communication complexity. This separation question can be equivalently phrased as proving randomized lower bounds for the \emph{Clique vs.~Independent Set} game~\cite{yannakakis91expressing}. This result now follows by lifting an analogous separation in query complexity, obtained in several papers~\cite{goos15randomized,ambainis16separations,aaronson16separations}; it was previously shown using the lifting theorem of \cite{goos16rectangles}, which requires a query lower bound in a model stronger than $\BPPdt$.

\paragraph{Approximate Nash equilibria.}
Babichenko and Rubinstein~\cite{babichenko17communication} showed a randomized communication lower bound for finding an approximate Nash equilibrium in a two-player game. Their approach was to show a lower bound for a certain query version of the $\PPAD$-complete $\Eol$ problem, and then lift this lower bound into communication complexity using~\cite{goos16rectangles}. However, as in the above Clique vs.~Independent Set result, the application of~\cite{goos16rectangles} here requires that the query lower bound is established for a model stronger than $\BPPdt$, which required some additional busywork. Our lifting theorem can be used to streamline their proof.

\bigskip
\subsection*{Acknowledgements}
Thanks to Shalev Ben-David and Robin Kothari for quantum references. Thanks to Anurag Anshu, Rahul Jain, Raghu Meka, Aviad Rubinstein, and Henry Yuen for discussions.

\bigskip

\DeclareUrlCommand{\Doi}{\urlstyle{sf}}
\renewcommand{\path}[1]{\small\Doi{#1}}
\renewcommand{\url}[1]{\href{#1}{\small\Doi{#1}}}
\newcommand{\etalchar}[1]{$^{#1}$}

\end{document}